\documentclass[sigplan,10pt,nonacm,screen]{acmart}
\usepackage{amsmath,mathtools}
\usepackage{booktabs}
\usepackage{tabularx}
\usepackage{enumitem}
\usepackage{algorithm}
\usepackage[noend]{algpseudocode}
\usepackage{xspace}
\newcommand{\alg}{\textsc{SeniorLock}\xspace}
\newcommand{\objalg}{\textsc{SeniorObj}\xspace}
\newcommand{\Next}{\mathit{Next}}
\newcommand{\Cand}{\mathit{Cand}}
\newcommand{\Owner}{\mathit{Owner}}

\newcommand{\INF}{\mathsf{INF}}
\newcommand{\ticket}{\mathit{ticket}}
\newcommand{\pc}{\kappa}
\newcommand{\nullowner}{\bot}

\newtheorem{definition}{Definition}
\newtheorem{lemma}{Lemma}
\newtheorem{theorem}{Theorem}
\newtheorem{corollary}{Corollary}
\newtheorem{proposition}{Proposition}
\theoremstyle{remark}

\begin{document}

\title{Wait-Free Locks Should Not Fear Later Arrivals}

\author{Tong Che}
\affiliation{%
  \institution{NVIDIA Research}
  \country{USA}
}
\email{tongc@nvidia.com}

\begin{abstract}
Helping seems to make a lock wait-free: wrap the critical section in an
idempotent thunk that any process can finish once the holder stalls.
Yet helping protects the system, not the call.
An overwritable candidate lets later requests bump one another in
sequence, so a call can be forced to help newcomer after newcomer and
never return, while point contention never exceeds two.

We ask whether a call can instead be charged only for the requests
active when it takes its ticket, never for what arrives afterward.
We show that the answer is yes.
\alg is a deterministic helpable thunk lock in which a call with
ticket-time seniority $\beta$ finishes in $O((\beta+1)(T+1))$ local
shared-memory steps, where $T$ bounds one thunk's cost, independently
of later invocations.
We call this guarantee \emph{retrospective wait-freedom}.

The same lock doubles as a universal construction we call \objalg:
it turns any deterministic sequential object whose operations are
bounded, concurrently idempotent thunks into a retrospective
wait-free one, with no copying of its representation and, when no
senior is active, at essentially the native cost of the operation.
\end{abstract}


\maketitle
\pagestyle{plain}

\section{Introduction}

Helping seems to turn a lock into a wait-free object.
Package the critical section as an idempotent thunk; if its holder stalls,
another process finishes the thunk and releases the lock
\cite{TurekShashaPrakash1992,Barnes1993,CensorHillelPetrankTimnat2015,
BenDavidBlellochWei2022}.

The intuition confuses global progress with the progress of one call.
The cited deterministic thunk locks are lock-free, while randomized
wait-free alternatives give probabilistic contention-dependent bounds
\cite{BenDavidBlelloch2022,AeiniEtAl2026}.
In a single-candidate deterministic design, a newly arrived request can
redirect the victim's every step.
Only the victim and the current redirecting call need be active, so point
contention stays at two while the victim takes infinitely many steps without
returning (Section~\ref{sec:design}).

The gap is temporal.
The calls already active at a chosen step are finite; the stream arriving
afterward need not be.
Can a deterministic helpable thunk lock charge a call only for the former?

We present \alg.
Each call takes an increasing ticket and publishes itself with an atomic
minimum update, so that no larger-ticket call can later replace its
candidate.
Its \emph{seniors} are the smaller-ticket calls already active when its
ticket is assigned; a later arrival always receives a larger ticket, but a
delayed earlier invocation may receive one too.
Repeated publication restores this guard after a senior completes, and a
versioned owner compresses stale larger-ticket installations.

We call the resulting guarantee \emph{retrospective wait-freedom}: an
operation's worst-case bound may depend on its seniors when its key is
fixed, but not on later invocations.
To our knowledge, \alg is the first deterministic helpable thunk lock in this
model with such a bound.
If $\beta_D$ is the ticket-time seniority of $D$, then
$\mathsf{Lock}(D)$ uses $O(\beta_D+1)$ management steps and
$O((\beta_D+1)(T+1))$ total local steps under every asynchronous schedule.
Since $\beta_D\le\pc_D-1$, the result is also point-contention adaptive.

\alg is also universal.
One instance linearizably turns any deterministic sequential object whose
private representation is accessed by thunks satisfying the uniform
$T$-step service contract into \objalg, a retrospective wait-free object
that never copies that representation into one atomic state transition
(Section~\ref{sec:universal}).
Prior generalized compare-and-swap work obtains the fixed-senior bound for
atomic state transitions \cite{HadzilacosThiessenToueg2026}.
\alg extends it to bounded, multi-step, effectful thunks.

The bound relies on resettable min-CAS as one base-object step.
For fixed $N$, implementing this primitive from unbounded
read-write-conditional registers has tight $\Theta(N)$ black-box location
complexity.
With infinitely many arrivals, no read/write/CAS implementation has constant
quiescent location complexity (Section~\ref{sec:separation}).

\paragraph{Contributions.}
\begin{enumerate}[leftmargin=2em]
  \item We formulate \emph{retrospective wait-freedom}, which freezes an
  operation's contention-dependent bound when its key is fixed, and present
  \alg, a deterministic helpable thunk lock with
  $O((\beta+1)(T+1))$ local-step complexity
  (Sections~\ref{sec:model}--\ref{sec:algorithm}).

  \item We prove a finite charge-set service theorem that applies beyond FAI
  ticket order, and give counterexamples showing why one-time publication and
  unversioned ownership lose the bound
  (Sections~\ref{sec:progress} and~\ref{sec:ablation}).

  \item The same lock is also universal.
  We derive \objalg, an \emph{effectful retrospective universal
  construction} that turns any deterministic sequential object whose
  operations satisfy the uniform concurrent-idempotence service contract
  into a retrospective wait-free one, with no representation-copying step
  (Section~\ref{sec:universal}).

  \item Resettable min-CAS is not a free primitive.
  For a fixed population of $N$ processes, we prove a tight
  $\Theta(N)$ location bound for any long-lived black-box implementation
  from read-write-conditional registers, even with unbounded values, and
  that no read/write/CAS implementation keeps constant quiescent location
  complexity once arrivals are unbounded
  (Section~\ref{sec:separation}).
\end{enumerate}

\section{Model and lock semantics}\label{sec:model}

\paragraph{Shared-memory model.}
There are $N$ deterministic asynchronous processes, each with at most one
outstanding $\mathsf{Lock}$ call, communicating through linearizable atomic
objects \cite{HerlihyWing1990,AttiyaWelch2004}.
One base-object access is one step, and the scheduler is unrestricted
\cite{Herlihy1991}.
The main bound is independent of $N$; Section~\ref{sec:separation} fixes $N$
for its register-space classification.

\paragraph{Descriptors and idempotent thunks.}
Every lock call allocates a fresh descriptor
\[
  D=(D.\mathit{id},D.\ticket,D.\mathit{done},
     D.\mathit{thunk},D.\mathit{result}).
\]
The identity and ticket are immutable and globally unique.
$D.\mathit{done}$ changes monotonically from false to true; reads and the one
false-to-true update are linearizable atomic accesses.
A \emph{run} of a thunk is the sequence of shared-data steps taken by one
process while executing or helping it.
For an execution $E$, let $E|D$ be the projection to steps belonging to runs
of $D$'s thunk.
We require the following observational form of concurrent idempotence from
Ben-David, Blelloch, and Wei \cite{BenDavidBlellochWei2022}.

\begin{definition}[Concurrent idempotence]\label{def:idempotence}
For every valid execution $E$ containing arbitrarily interleaved runs of one
thunk together with arbitrary other shared-data steps, there is a subsequence
$E'$ of $E|D$ such that replacing $E|D$ by $E'$ leaves a valid history
consistent with one run, preserves all non-thunk steps and responses, and is
indistinguishable under every continuation.
If a run finishes, the last step of the first finished run is the end of
$E'$.
\end{definition}

Thus the physical runs contain one logical copy whose \emph{logical execution
interval} runs from the first thunk step to the end of the first finished
run; subsequent steps are noneffectual.

Every $D.\mathsf{run}()$ call, independently of helper count, takes at most
$T$ local shared-memory steps to observe done or finish a run and set done;
it does not return while $D$ remains active.
If the thunk returns a value, its idempotence log records that value in the
descriptor before publishing done, and \textsc{Lock} returns the recorded
value.
The bound $T$ includes the descriptor-local logging needed to make the thunk
idempotent, but excludes the lock-management operations shown in
Algorithm~\ref{alg:seniorlock}.

\begin{definition}[Helpable thunk lock]
A call $\mathsf{Lock}(D)$ returns only after $D.\mathit{done}$ is true.
The logical execution interval of $D$ lies between the invocation and
response of $\mathsf{Lock}(D)$.
Logical execution intervals of distinct descriptors do not overlap.
\end{definition}

We abbreviate $\mathsf{Lock}(D,f)$ to $\mathsf{Lock}(D)$.
The interface serializes logical thunk intervals rather than physical
acquire/release ownership; late physical runs are noneffectual.

\paragraph{Contention and complexity.}
A lock call is \emph{active} from invocation through response; we use the
standard point-contention measure \cite{AttiyaFouren2003}.
For $\mathsf{Lock}(D)$, its maximum point contention is
\[
  \pc_D=\max_s
  \left|\{E:\mathsf{Lock}(E)\text{ is active at configuration }s\}\right|,
\]
where the maximum ranges over configurations during $\mathsf{Lock}(D)$.

\paragraph{Retrospective wait-freedom.}
Point contention measures overlap at a configuration.
The following object-independent property instead fixes the operations that
may enter a contention-dependent bound at a designated step.

\begin{definition}[Retrospective wait-freedom]\label{def:seniority}
Consider any concurrent object whose operations each receive a unique
\emph{key} from a total order, fixed no later than some designated step
$\sigma_{op}$ during the operation's execution.
The smaller-key operations active immediately after $\sigma_{op}$ are
$op$'s \emph{seniors}; its \emph{seniority} is the size of that set.
An implementation is \emph{retrospective wait-free} if there is a function
$f$ such that every operation completes within $O(f(\text{seniority}))$
steps of its invoking process under every asynchronous schedule, independent
of operations invoked after $\sigma_{op}$.
\end{definition}

An operation's seniors are a subset of the operations active at
$\sigma_{op}$, so retrospective wait-freedom implies point-contention
adaptivity and additionally makes the bound independent of invocations after
$\sigma_{op}$.

We instantiate Definition~\ref{def:seniority} for \alg's $\mathsf{Lock}$
operation, using the FAI ticket as the key.

\paragraph{\alg seniority.}
For the call $\mathsf{Lock}(D)$, let $\sigma_D=\tau_D$ be the linearization
point of $D$'s FAI and let $t_D$ be the returned ticket.
Define
\[
  \begin{aligned}
  \mathsf{Seniors}(D)=\{E:\;&E\text{ is active immediately after }\tau_D,\\
                            &\tau_E<\tau_D\}.
  \end{aligned}
\]
$D$'s \emph{seniority} is
$\beta_D=|\mathsf{Seniors}(D)|$.

Here $\tau_E$ must exist: an invoked call still paused before its FAI is not
a senior.
Because FAI tickets increase, $\tau_E<\tau_D$ is equivalent to
$t_E<t_D$ for ticketed calls.
Every later arrival is therefore a larger-ticket call; the converse may fail
for a call invoked before $\tau_D$ but delayed before its FAI.
All calls represented in $\mathsf{Seniors}(D)$ overlap
$\mathsf{Lock}(D)$ at $\tau_D$, so
$\beta_D\le\pc_D-1$.
The inequality can be strict by an arbitrary amount: after $\tau_D$, any
number of larger-ticket calls may become active before $D$'s first minimum
update.
Our main bound depends on $\beta_D$, not on that later peak.

This differs from \emph{interval contention}, which counts every call that
overlaps $\mathsf{Lock}(D)$ at any time.
Interval contention can be unbounded while point contention is two: one
later call may finish before the next begins.
All complexity bounds count shared-memory steps by the process invoking
$\mathsf{Lock}(D)$.
A \emph{management step} is any shared-memory access shown in
Algorithm~\ref{alg:seniorlock} outside a thunk run.
The bounds do not cover aggregate system work or wall-clock time under an
asynchronous scheduler.

\paragraph{Base objects.}
The lock contains three abstract objects.
\begin{itemize}[leftmargin=1.5em]
  \item $\Next$ is an unbounded fetch-and-increment (FAI) counter.
  \item $\Cand$ is a resettable min-CAS object.
  It stores either $\INF$ or a descriptor token, ordered by the descriptor's
  immutable ticket, and supports atomic
  $\mathsf{Read}$, $\mathsf{CAS}$, and
  $\mathsf{MinUpdate}(x):C\gets\min(C,x)$.
  The initial value is $\INF$.
  \item $\Owner$ is a CAS object containing $(p,g)$, where $p$ is a
  descriptor pointer or $\nullowner$ and $g$ is an unbounded generation.
  Initially $\Owner=(\nullowner,0)$.
\end{itemize}
Tickets are unique, so they totally order descriptors.
The same candidate object also supports exact-token CAS reset, so its
minimum-update monotonicity holds between successful resets.

The theorem assumes unbounded tickets and generations and fresh descriptors;
its three-object count is lock-resident only
(Section~\ref{sec:scope}).

\section{From replacement to a ticket guard}\label{sec:design}

\paragraph{Destructive replacement.}
The single-overwrite candidate pattern admits unbounded cumulative bypass at
constant point contention.
Consider a single overwrite slot and a victim $D$.
In round $i$, let $D$ publish and pause; start one later call $Y_i$, overwrite
$D$, and let $D$ help $Y_i$ finish.
Return $Y_i$ before starting $Y_{i+1}$.
The victim takes steps in every round and never completes, but at every
configuration at most $D$ and one $Y_i$ are active.
The complexity driver is cumulative bypass, not point contention.

\paragraph{A stranded descriptor-side obligation.}
Adding only a descriptor-local selected bit to this pattern still permits a
helper to mark a descriptor after the candidate has changed.
That mark is persistent only in the descriptor, not in the arbitration state.
For example, let a smaller-ticket request $E$ replace $D$ after a helper reads
$\Cand=D$ but before it marks $D$.
After $E$ finishes and clears the candidate, the stale helper can mark $D$ as
selected.
Now $D$ is selected but absent from $\Cand$.
Sequential larger-ticket calls can repeatedly install themselves immediately
before $D$'s owner CAS, and $D$ can be forced to help each of them.
Again the schedule has point contention two.
Thus the local bit does not impose a shared order; a pending request needs a
\emph{shared ordering constraint}.

\paragraph{An unversioned installation cohort.}
Even with a monotone candidate, selection and owner installation are separate
atomic steps.
Suppose $r$ helpers read the same candidate $x$ and an empty, untagged owner,
then all pause before CAS.
After the first helper installs and completes $x$, every delayed CAS can
successfully install the now-done $x$ again whenever the owner returns to
null.
All $r$ calls remain simultaneously active, so this does not violate a point
contention bound of $r+1$.
Repeating such a cohort after several smaller-ticket completions gives a
quadratic number of victim-visible tenures; Section~\ref{sec:ablation}
formalizes the schedule.

\paragraph{The ticket guard.}
Minimum update supplies the shared ordering constraint.
Suppose $D$ has ticket $t_D$ and performs
$\Cand.\mathsf{MinUpdate}(D)$.
Immediately afterward,
\[
  \Cand\ne\INF
  \quad\text{and}\quad
  \Cand.\ticket\le t_D. \tag{1}\label{eq:guard}
\]
No later request can falsify~\eqref{eq:guard}.
Only a reset of the exact current candidate can do so.
Our algorithm permits such a reset only after that candidate is done.
Any smaller-ticket candidate is either from a call active at $\tau_D$ or is
the one possible completed-token residue, which contributes one cleanup
charge.
Before~\eqref{eq:guard} is established, several active helpers may have read a
larger-ticket candidate and paused before owner installation.
They may resume after the guard.
If owner emptiness were represented by an untagged null pointer, all of those
stale installations could succeed one after another as the owner repeatedly
returned to null.
In \alg they all expect the same generation.
At most one succeeds; when that tenure ends, the generation changes and the
rest fail.
No new larger-ticket installation can be prepared while the guard holds.
Repeated self-publication restores the guard after a legitimate smaller-ticket
completion, while direct owner installation uses the successful CAS itself
as the unique selection event.

The mechanism accepts any total order for which the call has a finite charge
set.
With increasing FAI tickets, the senior set itself is a valid charge
set (Section~\ref{sec:progress}).

\section{The \alg algorithm}\label{sec:algorithm}

Algorithm~\ref{alg:seniorlock} gives the construction.
The requester initializes its descriptor, obtains a ticket, and repeatedly
publishes itself before helping.
The publication is unconditional while the descriptor is not done.
A delayed smaller-ticket request can replace the current candidate token, so
every still-active requester republishes to restore its ticket guard after
that request completes.

\begin{algorithm}[t]
\caption{\alg for one helpable thunk lock.}\label{alg:seniorlock}
\begin{algorithmic}[1]
\Procedure{Lock}{$D,f$}
  \State $D.\mathit{thunk}\gets f$;\quad
         $D.\mathit{done}\gets\mathsf{false}$;\quad
         $D.\mathit{result}\gets\bot$
  \State $D.\ticket\gets\Next.\mathsf{FAI}()$
  \While{$\neg D.\mathit{done}$}
    \State $\Cand.\mathsf{MinUpdate}(D)$
    \State \Call{Assist}{}
  \EndWhile
  \State \Return $D.\mathit{result}$
\EndProcedure
\Statex
\Procedure{Assist}{}
  \State $(w,g)\gets\Owner.\mathsf{Read}()$
  \If{$w\ne\nullowner$}
    \State \Call{HelpOwner}{$w,g$}
    \State \Return
  \EndIf
  \State $x\gets\Cand.\mathsf{Read}()$
  \If{$x=\INF$}
    \State \Return
  \EndIf
  \If{$x.\mathit{done}$}
    \State $\Cand.\mathsf{CAS}(x,\INF)$
    \State \Return
  \EndIf
  \If{$\Owner.\mathsf{CAS}((\nullowner,g),(x,g))$}
    \State \Call{HelpOwner}{$x,g$}
  \EndIf
\EndProcedure
\Statex
\Procedure{HelpOwner}{$x,g$}
  \If{$\Owner.\mathsf{Read}()\ne(x,g)$}
    \State \Return
  \EndIf
  \If{$\neg x.\mathit{done}$}
    \State $x.\mathsf{run}()$
  \EndIf
  \If{$x.\mathit{done}$}
    \State $\Cand.\mathsf{CAS}(x,\INF)$
    \State $\Owner.\mathsf{CAS}((x,g),(\nullowner,g+1))$
  \EndIf
\EndProcedure
\end{algorithmic}
\end{algorithm}

\subsection{Critical interleavings}\label{sec:races}

\paragraph{The done test and generation form a handshake.}
The required order in \textsc{Assist} is
\[
  \text{read empty owner }(\nullowner,g)
  \ \prec\ \text{read }x.\mathit{done}
  \ \prec\ \text{owner CAS}.
\]
Assume the middle read returns false but $x$ becomes done before the CAS.
The only way to make $x$ done is through a tenure that must be released before
the owner can again be empty.
That release writes generation at least $g+1$, so the delayed CAS expecting
$g$ fails.
If the done test were moved before the owner read, a helper could read
false, pause while another helper completes $x$, then read the \emph{new}
empty generation and install the completed descriptor.
This owner-read/done-read/CAS order is required for correctness.

\paragraph{Other stale observations.}
Three further schedules -- a stale candidate becoming owner, a late
physical run resuming after revalidation, and a new candidate published just
after a reset -- are handled by the same combination of exact-token
cleanup and concurrent idempotence.
See Appendix~\ref{app:safety} (``Critical interleavings'') for all four
schedules.

\section{Safety}\label{sec:safety}

An \emph{owner tenure} $(x,g)$ begins when a CAS changes
$(\nullowner,g)$ to $(x,g)$ and ends when a CAS changes $(x,g)$ to
$(\nullowner,g+1)$.
Table~\ref{tab:authorization} summarizes the state changes that require
proof obligations.

\begin{table}[t]
\caption{Authorization for destructive state changes.}
\label{tab:authorization}
\begin{tabularx}{\columnwidth}{@{}lX@{}}
\toprule
Change & Required evidence \\
\midrule
$\Owner:(\nullowner,g)\!\to\!(x,g)$
  & Empty owner was read in generation $g$; afterward $x$ was read not done. \\
$\Cand:x\!\to\!\INF$
  & $x$ was observed done; the CAS matches the exact descriptor token. \\
$\Owner:(x,g)\!\to\!(\nullowner,g+1)$
  & $x$ was observed done while the same owner tenure was still installed. \\
\bottomrule
\end{tabularx}
\end{table}

\begin{lemma}[Generation uniqueness]\label{lem:generation}
For every generation $g$, at most one descriptor is installed from
$(\nullowner,g)$.
After generation $g$ ends, every delayed installation expecting
$(\nullowner,g)$ fails.
\end{lemma}

\begin{lemma}[Release after completion]\label{lem:release}
If tenure $(x,g)$ ends, then $x.\mathit{done}$ was true before the releasing
CAS.
\end{lemma}

\begin{lemma}[Candidate responsibility]\label{lem:candidate}
If $\Cand=x$ and $x$ is not done, no process can reset $\Cand$ from $x$ to
$\INF$.
A minimum update can replace $x$ only by a smaller-key descriptor.
\end{lemma}

\begin{lemma}[One tenure per descriptor]\label{lem:one-tenure}
Every descriptor is installed as owner at most once.
\end{lemma}

\paragraph{Safety invariants.}
CAS on the complete owner pair gives generation uniqueness; the release code
first observes the monotone done flag.
Candidate reset likewise follows a done observation and compares the exact
token.
After one tenure, pre-release installation snapshots expect the old
generation, while post-release snapshots observe done, proving one tenure.
See Appendix~\ref{app:safety} (``Safety invariants'') for the detailed
proofs.

\begin{theorem}[Safety]\label{thm:safety}
\alg implements a helpable thunk lock.
\end{theorem}

\begin{proof}[Proof sketch]
Every effectual run follows validation of one owner tenure.
A later tenure cannot begin before the prior descriptor is done; if a
validated helper resumes after release, concurrent idempotence makes that run
noneffectual.
Thus generation order serializes logical thunk intervals.
Publication and completion place each such interval within its call, and the
generation order respects real-time precedence.
\end{proof}

\section{Finite charge-set progress}\label{sec:progress}

The charging argument accepts any total order for which the operation has a
finite charge set.
It bounds a call by a finite \emph{charge set} containing every smaller-key
operation that can affect it.
With increasing FAI tickets, the senior set is a valid charge set,
yielding retrospective wait-freedom (Definition~\ref{def:seniority}).
Section~\ref{sec:separation} instead uses static keys, which retain a finite
wait-free bound but may charge later arrivals.
The progress argument uses only monotone completion and bounded service; it
is independent of the thunk semantics in Section~\ref{sec:model}.

\begin{definition}[Ordered service instance]\label{def:ordered-service}
An ordered service instance follows the control flow of
Algorithm~\ref{alg:seniorlock}, except that every operation $x$ receives an
arbitrary unique immutable key $\rho(x)$ before its first publication.
The done flag may be replaced by a monotone predicate $\mathsf{Done}(x)$ that
can first become true only through service of owner $x$.
After revalidating tenure $(x,g)$, a matching service call makes or observes
$\mathsf{Done}(x)$ within at most $T$ local shared-memory steps.
Every pending requester repeatedly applies
$\Cand.\mathsf{MinUpdate}(x)$ before calling \textsc{Assist}.
\end{definition}

Fix an execution $\alpha$ and a descriptor $D$ in it.
Let $\sigma_D$ be the point at which $D$'s key is fixed and write
$d=\rho(D)$.
The \emph{pending suffix} of $D$ is the part of $\alpha$ after $\sigma_D$ and
before $D$ becomes done, or the entire remaining execution if it never does.

\begin{definition}[Charge set]\label{def:charge-set}
A finite set $\mathcal S_D(\alpha)$ is a \emph{charge set} for $D$ in
$\alpha$ if every descriptor $x\ne D$ with $\rho(x)<d$ belongs to
$\mathcal S_D(\alpha)$ whenever, during $D$'s pending suffix, at least one of
the following holds:
\begin{enumerate}[label=(\roman*),leftmargin=2em]
  \item $x$ is not done at $\sigma_D$ and is named there by $\Cand$ or
        $\Owner$;
  \item the requester of $x$ publishes $x$ after $\sigma_D$;
  \item $x$ begins an owner tenure after $\sigma_D$.
\end{enumerate}
Set $m=|\mathcal S_D(\alpha)|$ and abbreviate the set by $\mathcal S_D$ when
$\alpha$ is fixed.
\end{definition}

A smaller-key done token already in $\Cand$ at $\sigma_D$ that is never
republished need not belong to $\mathcal S_D$: it can be removed once.
Likewise, one owner tenure already current at $\sigma_D$ is treated as an
initial residue.
These are the two constant terms in the charging bounds below.
A tenure is \emph{encountered} by $D$ when a step of $D$ reads its owner pair
or an owner transition in that tenure makes $D$'s CAS/revalidation fail.
For a CAS prepared from an empty-owner read, charge failure to the first
tenure that makes that expected pair unreachable and to the interval
containing the read, even if the delayed CAS executes in a later interval.

\begin{definition}[Guarded configuration]
A configuration is $D$-\emph{guarded} when
\[
  \Cand\ne\INF
  \quad\text{and}\quad
  \rho(\Cand)\le d.
\]
A maximal interval of guarded configurations is a $D$-guarded interval.
\end{definition}

\begin{lemma}[Guard establishment]\label{lem:establish}
\leavevmode\par
Every $\Cand.\mathsf{MinUpdate}(D)$ establishes the $D$-guard unless $D$ is
already done.
While the guard holds, no larger-key descriptor can become the value of
$\Cand$.
\end{lemma}

\begin{lemma}[Smaller-key chain compression]\label{lem:smaller-chain}
Smaller-key publications never destroy the guard or split a guarded interval.
Within one guarded interval, changes of $\Cand$ caused by first-time
smaller-key publications form a strictly decreasing key chain of length at
most $m$.
\end{lemma}

\begin{proof}[Proof sketch]
Every effective smaller-key publication strictly decreases the candidate key
and comes from the charge set of size $m$.
Repeated publications do not extend this chain, so they share one guarded
interval rather than creating $m$ contention-sized phases.
\end{proof}

\begin{lemma}[Opening the guard]\label{lem:open}
Before $D$ is done, every transition from a $D$-guarded configuration to an
unguarded configuration is associated with a completed descriptor in
$\mathcal S_D$, except for at most one done descriptor already resident in
$\Cand$ at $\sigma_D$.
Each descriptor in $\mathcal S_D$ is associated with at most two such
transitions.
\end{lemma}

\begin{proof}[Proof sketch]
The guard can be destroyed only by a successful reset of the exact current
candidate after that candidate is done.
Except for the initial residue, a smaller-key candidate belongs to
$\mathcal S_D$ by the charge-set definition.
It is removed once after completion and can be reinserted only by one
already-started requester iteration, giving multiplicity two.
\end{proof}

\begin{lemma}[Done-candidate episodes]\label{lem:done-episodes}
For each descriptor $x$, there are at most two maximal intervals in which
$\Cand=x$ and $x$ is done.
Every failed cleanup CAS by $D$ after reading such an $x$ is charged to one
of these intervals.
\end{lemma}

\begin{lemma}[One larger-key bypass per guard]\label{lem:one-larger}
During one maximal $D$-guarded interval, the requester encounters at most one
larger-key owner tenure, including a tenure already current when the interval
begins.
\end{lemma}

\begin{proof}[Proof sketch]
By Lemma~\ref{lem:establish}, no larger-key descriptor can become the
candidate while the interval is guarded.
If a larger-key owner is current at interval start, it is the one bypass; its
release invalidates all earlier snapshots, and the guard prevents new ones.
Otherwise every prepared larger-key installation expects the same empty
generation, so one succeeds and its release invalidates the cohort.
\end{proof}

\begin{lemma}[One larger-key tail bypass]\label{lem:tail}
After the guard is opened and before $D$'s next minimum update, the process
executing $\mathsf{Lock}(D)$ incurs at most one larger-key obligation:
an encountered owner tenure or one done-candidate episode that it reads.
\end{lemma}

\begin{proof}[Proof sketch]
At an opening, $D$ is between iterations, between its minimum update and
owner read, in one occupied-owner suffix, in one empty-owner suffix, or in
the cleanup that caused the opening.
These cases respectively expose no obligation, one owner path, one owner
path, one mutually exclusive candidate-cleanup/owner path, or no further read
before return.
See Appendix~\ref{app:progress} (``One larger-key tail bypass'') for the
program-counter analysis.
\end{proof}

\begin{lemma}[Larger-key bypass compression]\label{lem:late-compression}
Before $D$ becomes done, it incurs at most $4m+3$ larger-key bypasses,
independent of how many larger-key calls are invoked.
\end{lemma}

\begin{proof}
Let $h$ be the number of guard openings.
Lemma~\ref{lem:open} gives $h\le2m+1$, including the initial done-candidate
residue.
The $h+1$ guarded intervals and $h$ tails each contribute at most one
larger-key bypass, for $(h+1)+h\le4m+3$.
\end{proof}

\begin{lemma}[Bounded encountered tenures]\label{lem:tenures}
Before $D$ becomes done, its requesting process helps
$O(m+1)$ distinct owner tenures.
\end{lemma}

\begin{lemma}[One local service per tenure]\label{lem:one-service}
For each encountered tenure, the requester follows a matching
\textsc{HelpOwner} path at most once and invokes service at most once within
that path (zero times if done is already true).
\end{lemma}

\begin{lemma}[Terminal suffix]\label{lem:terminal}
After $D$ first becomes done and before its request responds, its invoking
process performs only $O(1)$ management steps and at most one service call.
\end{lemma}

\begin{table}[t]
\caption{Explicit charging for a fixed descriptor $D$.}
\label{tab:charges}
\begin{tabularx}{\columnwidth}{@{}Xll@{}}
\toprule
Victim-visible event & Charge & Bound \\
\midrule
Guard opening & done smaller-key descriptor & $h\le 2m+1$ \\
Larger-key tenure inside a guard & guarded interval & $\le h+1$ \\
Larger-key obligation in an unguarded tail & preceding opening & $\le h$ \\
Smaller-key owner tenure & its descriptor & $\le m$ \\
Victim/residual tenure & constant term & $\le 2$ \\
\midrule
Owner/bypass subtotal & last four rows & $\le 5m+5$ \\
\bottomrule
\end{tabularx}
\end{table}

Including one initially resident done candidate, Table~\ref{tab:charges}
totals at most $7m+6$ charges and hence $O(m+1)$.

\begin{lemma}[Management-step charging]\label{lem:management}
The requester of $D$ performs $O(m+1)$ management steps before $D$ is done.
\end{lemma}

\begin{proof}[Proof sketch]
Partition $D$'s precompletion iterations by their \textsc{Assist} outcome.
An occupied-owner path or owner-CAS failure is charged to the tenure that
changed the observed owner; one tenure receives at most a failed installation
and one subsequent owner path, which ends or releases it.
An $\INF$ read is charged to its intervening opening, once before
republication.
A done-candidate path is charged to its maximal done episode; sequential $D$
attempts cleanup at most once per episode.
Table~\ref{tab:charges} bounds all three charge sets by $O(m+1)$.
See Appendix~\ref{app:progress} (``Management-step charging'') for the
complete case analysis.
\end{proof}

\begin{theorem}[Finite charge-set service]\label{thm:finite-charge}
For every execution $\alpha$, ordered service descriptor $D$, and charge set
$\mathcal S_D(\alpha)$, after $\sigma_D$ and until response the invoking
process performs at most
$O(|\mathcal S_D|+1)$ management steps and
\[
  O((|\mathcal S_D|+1)(T+1))
\]
local shared-memory steps.
If initialization and key assignment use $O(1)$ shared-memory steps, the same
asymptotic bounds hold for the full invocation.
\end{theorem}

\begin{proof}
Table~\ref{tab:charges} and Lemma~\ref{lem:management} give
$O(|\mathcal S_D|+1)$ management steps and service calls.
Each service costs at most $T$, and the terminal suffix costs $O(T+1)$.
\end{proof}

\subsection{Increasing-ticket instantiation}

\begin{lemma}[FAI charge set]\label{lem:fai-charge}
For \alg, $\mathsf{Seniors}(D)$ is a charge set for $D$.
Its size is $\beta_D\le\pc_D-1$.
\end{lemma}

\begin{proof}[Proof sketch]
Increasing FAI prevents a call invoked after $\tau_D$ from obtaining a
smaller ticket.
A smaller-ticket call that returned before $\tau_D$ cannot publish again,
and every delayed installation other than the current tenure expects an
obsolete owner generation.
Thus every chargeable smaller-ticket call is a senior of $D$.
All such calls overlap $\mathsf{Lock}(D)$ at $\tau_D$, giving the cardinality
bound.
See Appendix~\ref{app:progress} (``FAI senior cover'') for the residue
argument.
\end{proof}

\begin{theorem}[Retrospective wait-freedom]\label{thm:waitfree}
For every call $\mathsf{Lock}(D)$ and every asynchronous schedule,
$\mathsf{Lock}(D)$ returns after $O(\beta_D+1)$ lock-management steps and
$O((\beta_D+1)(T+1))$ local shared-memory steps.
\end{theorem}

\begin{proof}
Apply Theorem~\ref{thm:finite-charge} to the senior set from
Lemma~\ref{lem:fai-charge}, whose size is $\beta_D$.
The service procedure is $x.\mathsf{run}()$, with the $T$-step contract from
Section~\ref{sec:model}; descriptor initialization and the FAI consume
$O(1)$ additional shared-memory steps.
\end{proof}

\begin{corollary}[Point-contention adaptivity]\label{cor:point}
Every call $\mathsf{Lock}(D)$ returns in $O(\pc_D)$ lock-management steps and
$O(\pc_D(T+1))$ local shared-memory steps.
\end{corollary}

\begin{proof}
By construction, $\beta_D+1\le\pc_D$.
\end{proof}

\begin{corollary}[Independence from later arrivals]\label{cor:late}
If no smaller-ticket call is active when $D$ obtains its ticket, then $D$
becomes done and $\mathsf{Lock}(D)$ returns in $O(T+1)$ local shared-memory
steps, independent of calls invoked after the ticket step.
\end{corollary}

\begin{corollary}[Lock-resident space]\label{cor:space}
\alg uses three lock-resident abstract base objects.
Every operation additionally allocates a descriptor and idempotence log.
\end{corollary}

\subsection{An effectful retrospective universal construction}
\label{sec:universal}

Let $\mathcal O$ be a deterministic sequential object with state space
$Q$, initial state $q_0$, and total transition function
\[
  \delta(q,op)=(q',r).
\]
Store a private shared representation $R$ of the current state, initially
representing $q_0$, and associate one \alg instance $L$ with $R$.
For an invocation of $op$, allocate a fresh descriptor $D$ and install a
thunk $f_{D,op}$ that executes the sequential code for $op$ against $R$ and
returns its response.
Define the concurrent implementation by
\[
  \mathsf{Invoke}_{\mathcal O}(op)
  =L.\mathsf{Lock}(D,f_{D,op}).
\]
All invocations on this instance of $\mathcal O$ use the same $L$.
This is a universal construction in the standard sense
\cite{Herlihy1991}; we call the resulting object \objalg.
\objalg never copies $R$ into an atomic state transition, unlike a
classical Herlihy-style universal construction.

\begin{theorem}[Effectful retrospective universality]
\label{thm:universal}
Let $\mathcal O$ be any deterministic sequential object for which every
invocation can be packaged as a concurrently idempotent thunk
$f_{D,op}$ satisfying the uniform, helper-count-independent $T$-step service
contract of Section~\ref{sec:model}.
One \alg instance implements $\mathcal O$ linearizably and retrospective
wait-free.
For an invocation with descriptor $D$ and seniority $\beta_D$, its
invoking process performs
\[
  O(\beta_D+1)
  \quad\text{lock-management steps and}\quad
  O((\beta_D+1)(T+1))
\]
local shared-memory steps, independent of operations invoked after
$D$ obtains its ticket.
\end{theorem}

\begin{proof}[Proof sketch]
By Theorem~\ref{thm:safety}, completed logical thunk intervals are
nonoverlapping, lie within their invocations, and respect real time.
Concurrent idempotence replaces every interval's physical runs by one run of
the corresponding sequential operation, so interval order is a legal
linearization of $\mathcal O$ with the recorded responses.
Done pending calls can be completed with those responses; an unfinished
tenure cannot be followed by a later logical interval, so all other pending
calls can be omitted.
Each target operation is one $\mathsf{Lock}(D,f_{D,op})$ call, and
Theorem~\ref{thm:waitfree} gives the claimed bounds with the same ticket,
senior set, and $T$.
See Appendix~\ref{app:universal} for the full safety and completion
argument.
\end{proof}

The logging transformation of Ben-David, Blelloch, and Wei supplies
concurrently idempotent thunks for ABA-free code whose shared accesses use
the transformed primitives \cite{BenDavidBlellochWei2022}; whenever the
transformed operation has a uniform $T$ bound, it instantiates
Theorem~\ref{thm:universal}.
Universality concerns operation semantics.
Total space includes the target representation, live descriptors, and
idempotence logs; Corollary~\ref{cor:space} counts the lock-resident objects.

\section{Ablations: republication and versioning}\label{sec:ablation}

This section states counterexamples for one-time publication and unversioned
ownership; Appendix~\ref{app:ablation} gives the complete schedules.

\subsection{One-time publication is not wait-free}

Consider the variant in which a request executes minimum update only once and
then repeatedly calls \textsc{Assist}.

\begin{proposition}[Republication]\label{prop:republication}
The one-time-publication variant has an infinite execution with maximum point
contention two in which a victim takes infinitely many steps and never
finishes.
\end{proposition}

\begin{proof}[Proof sketch]
A delayed smaller-ticket publication removes $D$'s sole candidate token.
Thereafter, start one larger-ticket call $Y_i$ per round immediately before
$D$ reads the owner, let $D$ complete it, and return it before starting
$Y_{i+1}$.
The schedule repeats at point contention two.
\end{proof}

The same schedule explains why a descriptor-local selected bit does not
replace republication.
After a stale read marks an absent $D$ selected, the bit records an obligation
but imposes no order on later owner installations.
In \alg the owner CAS is the selection event, and every still-pending caller
restores the shared ticket guard first.

\subsection{An untagged owner admits quadratic bypass}

Let \textsc{Untagged-SeniorLock} be Algorithm~\ref{alg:seniorlock} with
$\Owner\in\{\nullowner\}\cup\{D\}$ and CAS from an untagged null pointer.
All other lines, including the done check, are unchanged.

\begin{proposition}[Stale-cohort lower bound]\label{prop:untagged}
For every $q\ge1$, \textsc{Untagged-SeniorLock} has an execution with point
contention at most $2q+1$ in which one victim performs
$\Omega(q^2)$ lock-management steps.
\end{proposition}

\begin{proof}[Proof sketch]
Use $q$ smaller-ticket calls to open the guard in $q$ phases and $q$
larger-ticket helpers to prepare an installation cohort before each
republication of $D$.
Without a generation, all remaining stale CAS operations install the same
completed descriptor in turn.
Phase $i$ contributes $q-i+1$ installations, so the victim encounters
\[
  q+(q-1)+\cdots+1=\frac{q(q+1)}2
\]
tenures at point contention at most $2q+1$.
\end{proof}

\paragraph{Algorithm-specific tightness.}
For \alg, let $b$ senior calls publish one at a time and schedule $D$ to
run each $T$-step thunk before its own.
Then $\beta_D=b$ and $D$ performs $(b+1)T$ local thunk steps, showing that
the service-work term is asymptotically tight for \alg.

\section{Implementation model and limitations}\label{sec:scope}

The theorem treats MinUpdate, FAI, and tagged CAS as linearizable base-object
steps and assumes the representation conditions below.

\paragraph{Primitive and hardware.}
Arm LDUMIN, RISC-V AMOMIN, and CUDA \texttt{atomicMin} expose integer minimum
RMWs on supported targets
\cite{ArmLDUMIN2024,RISCVUnprivileged2024,NVIDIACUDA2026}.
The C++26 interface may use a CAS loop \cite{GrantKozickiNorthover2024}, but
Proposition~\ref{prop:cas-loop} shows that substitution is not wait-free.
The $N+1$-register CAS upper bound classifies register count, not
retrospective running time.

\paragraph{Representation and reclamation.}
The abstract candidate is the lexicographic token
$(D.\ticket,D.\mathit{id})$, equivalently an integer ticket plus an
immutable ticket-to-descriptor registry.
For at most $M$ calls, tickets index an $O(M)$ registry and the owner pair
uses $2\lceil\log_2(M+1)\rceil$ bits.
The unbounded-execution theorem assumes nonwrapping tickets and generations;
bounded concurrent timestamps address bounded key spaces under unbounded
arrivals \cite{IsraeliLi1993,DolevShavit1997}.
A delayed helper can retain a descriptor after both objects drop it, so
long-lived reuse needs reclamation metadata, for which hazard pointers
\cite{Michael2004} are one standard mechanism.
Total memory also includes the registry, descriptors, thunk logs, and any
reclamation metadata excluded by Corollary~\ref{cor:space}.

\paragraph{Memory ordering.}
In a release/acquire model, descriptors and done are release-published and
acquire-read; owner installation is acquire-release and cleanup is release.
Generation changes alone do not publish thunk effects.

\paragraph{Limits.}
The interface serializes one always-executed thunk and excludes multi-lock
acquisition, abort, nesting, strict FCFS order, and physical ownership
intervals; tickets order only published candidates.

\section{A conditional-RMW space separation}\label{sec:separation}

Because the time bound counts MinUpdate as one base step, this section
characterizes the base-location cost of implementing resettable min-CAS from
conditional RMW registers.

A \emph{resettable min-CAS register} stores either $\INF$ or a uniquely keyed
token and supports
\[
  \begin{gathered}
  \mathsf{Read}(),\qquad \mathsf{CAS}(x,y),\\
  \mathsf{MinUpdate}(x): C\gets\min(C,x).
  \end{gathered}
\]
Algorithm~\ref{alg:seniorlock} uses CAS only for exact reset
$\mathsf{CAS}(x,\INF)$.
For fixed input $x$, $\mathsf{MinUpdate}(x)$ changes every state larger than
$x$.
In contrast, a conditional RMW has, for fixed input, only one current value
on which it changes the register.
Fich, Hendler, and Shavit call a register supporting reads, writes, and such
operations a \emph{read-write-conditional register}; CAS is the canonical
example \cite{FichHendlerShavit2004}.
Their model permits unbounded register values, as does this section.
By \emph{black-box implementation} we mean a standalone linearizable object:
clients use only its interface, and its internal base registers are disjoint
from client registers.

\subsection{Why the direct CAS loop is insufficient}

The sequential effect of $\mathsf{MinUpdate}$ can be reproduced by repeatedly
reading $c$ and returning if $c\le x$.
Otherwise it attempts
\[
  \mathsf{CAS}(c,x)
\]
and retries after failure.
The resulting operation is not wait-free.

\begin{proposition}[CAS-loop starvation]\label{prop:cas-loop}
\leavevmode\par
After replacing $\mathsf{MinUpdate}$ by this retry loop,
Algorithm~\ref{alg:seniorlock} has an execution with $\beta_D=0$ and maximum
point contention two in which one invocation by $D$ takes infinitely many
steps without finishing its first minimum update.
\end{proposition}

\begin{proof}
In round $i$, let $D$ read a larger-ticket candidate $Y_i$ and pause before
$\mathsf{CAS}(Y_i,D)$.
Complete $Y_i$, including its exact reset of the candidate to $\INF$; the CAS
by $D$ then fails.
Return $Y_i$, start a fresh larger-ticket call $Y_{i+1}$, let it publish, and
repeat.
Only $D$ and the current $Y_i$ are active.
All $Y_i$ obtain their tickets after $D$, so $\beta_D=0$, but the retry loop
of $D$ never returns.
\end{proof}

Thus direct retry-loop substitution does not preserve wait-freedom.
The next theorem gives a space lower bound for any wait-free black-box
implementation.

\subsection{Locally finite keys}

To attribute the separation to minimum update, the reduction uses static
locally finite keys instead of \alg's fetch-and-increment object.
For process $p$'s $k$th operation, fix the key
\[
  \rho(p,k)=\pi(p,k)
  =\frac{(p+k)(p+k+1)}2+k,
\]
the computable Cantor injection; each identity $(p,k)$ is used once.
For an operation $X$, define
\[
  L_X=\{Y:\rho(Y)<\rho(X)\}.
\]
This set contains future as well as present operations, but is finite:
$|L_X|\le\rho(X)$.

\begin{corollary}[Locally finite keys]\label{lem:finite-key}
In an ordered service instance using the static keys above, operation $X$
finishes after
\[
  O(|L_X|+1)
\]
management steps and
$O((|L_X|+1)(T+1))$ local shared-memory steps of its invoking process.
\end{corollary}

\begin{proof}
The set $L_X$ is a charge set by Definition~\ref{def:charge-set}, including
smaller-key operations that arrive after $X$.
Apply Theorem~\ref{thm:finite-charge}.
\end{proof}

The FAI instance restricts the relevant smaller-key set to seniors.
Static keys give up that adaptivity but retain wait-freedom.

\subsection{A long-lived counter from three registers}

Fix $N$ processes.
In addition to one resettable min-CAS register $\Cand$, use a tagged CAS
register $\Owner$, initially $(\nullowner,0)$, and one CAS register
$\mathit{Reply}$.
The latter stores the entire vector
\[
  \mathit{Reply}[p]=(\mathit{seq},\mathit{result})
\]
as one unbounded value; initially every sequence number is zero.
Process $p$ keeps a private invocation number $k_p$, initially zero.
Process identifiers are never reused and each process has at most one
outstanding FAI.
Inductively, immediately before its $k$th invocation,
$\mathit{Reply}[p].\mathit{seq}=k-1$; operation $k+1$ is not announced before
operation $k$ responds.

\begin{algorithm}[t]
\caption{A long-lived FAI counter from resettable min-CAS.}
\label{alg:mincas-counter}
\begin{algorithmic}[1]
\Procedure{FAI}{}
  \State $k_p\gets k_p+1$;\quad
         $X\gets(\rho(p,k_p),p,k_p)$
  \While{$\neg\Call{Done}{X}$}
    \State $\Cand.\mathsf{MinUpdate}(X)$
    \State \Call{Assist}{}
  \EndWhile
  \State \Return $\mathit{Reply}.\mathsf{Read}()[p].\mathit{result}$
\EndProcedure
\Statex
\Function{Done}{$X=(\_,p,k)$}
  \State $A\gets\mathit{Reply}.\mathsf{Read}()$
  \State \Return $A[p].\mathit{seq}\ge k$
\EndFunction
\Statex
\Procedure{Complete}{$X=(\_,p,k),g$}
  \Repeat
    \State $A\gets\mathit{Reply}.\mathsf{Read}()$
    \If{$A[p].\mathit{seq}\ge k$}
      \State \Return
    \EndIf
    \State $A'\gets A$ with
           $A'[p]\gets(k,g)$
  \Until{$\mathit{Reply}.\mathsf{CAS}(A,A')$ succeeds}
\EndProcedure
\end{algorithmic}
\end{algorithm}

The \textsc{Assist} and \textsc{HelpOwner} procedures are those in
Algorithm~\ref{alg:seniorlock}, after replacing every test of
$x.\mathit{done}$ by $\mathsf{Done}(x)$ and the call
$x.\mathsf{run}()$ by $\mathsf{Complete}(x,g)$.
Thus the counter client contains no FAI object.

\begin{lemma}[Bounded completion]\label{lem:counter-complete}
\leavevmode\par
Every call to $\mathsf{Complete}(X,g)$ takes $O(1)$ shared-memory steps.
\end{lemma}

\begin{proof}
No different tenure begins while $(X,g)$ is current.
Every delayed old helper expects a reply vector preceding the monotone update
that enabled its tenure's release, so its CAS fails.
Thus a CAS failure for unfinished $X$ means another helper completed $X$;
one reread returns.
Every successful reply CAS advances exactly one sequence field from $k-1$ to
$k$, so the full vector never recurs.
\end{proof}

\begin{theorem}[Counter client]\label{thm:mincas-counter}
One resettable min-CAS register and two CAS registers suffice for a
deterministic wait-free linearizable long-lived FAI counter.
The call with token $X$ takes $O(\rho(X)+1)$ abstract shared-memory steps.
\end{theorem}

\begin{proof}[Proof sketch]
The successful reply CAS for $(p,k)$ linearizes the operation and returns its
owner generation.
The sequencing invariant prevents a later invocation from making $(p,k)$
appear done or overwriting its result before response.
Successive tenures therefore return $0,1,2,\ldots$ in real-time-respecting
order.
Completion is constant-time by Lemma~\ref{lem:counter-complete}; locally
finite keys and Corollary~\ref{lem:finite-key} give wait-freedom.
\end{proof}

\subsection{The separation}

Fich--Hendler--Shavit Theorem~3.6 states that an $N$-process wait-free
implementation of a long-lived object in $\mathrm{Visible}(N)$ from
read-write-conditional registers uses at least $\lceil N/2\rceil$ registers
\cite{FichHendlerShavit2004}.
Their visible-write definition explicitly includes long-lived counters.

\begin{theorem}[Conditional-RMW space separation]
\label{thm:mincas-separation}
Let $s(N)$ be the number of shared base locations used by a wait-free
linearizable long-lived black-box resettable min-CAS implementation for $N$
processes using only read, write, and conditional RMW registers.
Count read/write-only locations as conditional registers with unused
operations.
Then
\[
  \begin{aligned}
  s(N)+2&\ge\left\lceil\frac N2\right\rceil,\\
  s(N)&\ge\left\lceil\frac N2\right\rceil-2=\Omega(N).
  \end{aligned}
\]
The result holds even when every register stores an unbounded value.
\end{theorem}

\begin{proof}
Substitution in Algorithm~\ref{alg:mincas-counter} gives a wait-free
long-lived counter using $s(N)+2$ registers.
It uses the same $N$ client processes and introduces no helper process.
The visible-object lower bound of Fich, Hendler, and Shavit requires at least
$\lceil N/2\rceil$ read-write-conditional registers even with unbounded
values \cite{FichHendlerShavit2004}.
\end{proof}

\subsection{A matching CAS upper bound}

For processes $0,\ldots,N-1$, each with at most one outstanding operation, use
$N$ single-writer registers $\mathit{Ann}[p]$, initially $\bot$, and one CAS
register
\[
  \mathit{State}=(C,\mathit{Reply}[0..N-1]),
\]
initially $(\INF,[(0,\bot)]^N)$.
All registers hold unbounded values.
Process $p$'s $k$th operation has key
\[
  \rho_N(p,k)=N(k-1)+p.
\]
For this subsection, write
$L_X^N=\{Y:\rho_N(Y)<\rho_N(X)\}$.

Process $p$ forms $X=(p,k,\rho_N(p,k),op)$.
It writes the immutable token $X$ to $\mathit{Ann}[p]$.
Each help attempt collects all announcements, reads $\mathit{State}$, selects
the minimum-key entry one sequence beyond its recorded reply, applies its
transition to $C$, and attempts one state CAS.
The caller repeats until its reply appears; pseudocode and proof appear in
Appendix~\ref{app:separation}.

\begin{theorem}[Linear CAS upper bound]\label{thm:mincas-upper}
The construction is a wait-free linearizable long-lived resettable min-CAS
from $N$ read-write registers and one CAS register; operation $X$ performs
\[
  O\bigl(N(|L_X^N|+N+1)\bigr)
\]
shared-memory steps.
\end{theorem}

\begin{proof}[Proof sketch]
Each successful $\mathit{State}$ CAS applies one eligible announcement,
advances one reply sequence, prevents duplicate application and logical ABA,
and is the linearization point.
After $X$ publishes, at most $N-1$ already-running collects omit it; every
fresh collect that finds $X$ pending selects $X$ or an operation in $L_X^N$.
Thus distinct successful CAS operations cause at most $|L_X^N|+N$ retries,
each costing $N+O(1)$ accesses, before one final read observes completion.
\end{proof}

\begin{corollary}[Asymptotically tight fixed-population classification]
\label{cor:mincas-tight}
Let $s_{\min}(N)$ be the minimum register count for a wait-free linearizable
long-lived resettable min-CAS from read-write-conditional registers with
unbounded values.
Then
\[
  \left\lceil\frac N2\right\rceil-2
  \ \le\ s_{\min}(N)\ \le\ N+1,
  \qquad
  s_{\min}(N)=\Theta(N).
\]
\end{corollary}

\begin{proof}
Combine Theorems~\ref{thm:mincas-separation} and~\ref{thm:mincas-upper};
a read-write register is a permitted special case of a
read-write-conditional register.
\end{proof}

In the fixed-population black-box model with unbounded values, resettable
min-CAS has $\Theta(N)$ base-location complexity over read-write-conditional
registers.
This metric counts base locations rather than bits; CAS remains universal.

\subsection{Infinite arrivals and GCAS}

The same client yields a consequence in the infinite-arrival model
\cite{MerrittTaubenfeld2013,PerrinMostefaouiBonin2020}.
Let process identifiers range over $\mathbb{N}$, retain
$\rho(p,k)=\pi(p,k)$, and encode $\mathit{Reply}$ as a finite sparse map in
one flat, by-value unbounded CAS register.
Map sequence entries increase monotonically and are never removed; tokens are
the values $(\rho,p,k)$, not references to retained descriptors.
Every finite execution stores a finite map, while the number of reachable
locations remains three at the abstract level.
Following B\'edin et al., quiescent complexity after $n$ operations is the
maximum number of locations reachable when exactly those operations have
completed and none is pending \cite{BedinEtAl2021}.
It counts locations, not the bit length of one location's value.

\begin{corollary}[Infinite-arrival separation]
\label{cor:mincas-infinite}
In the infinite-arrival model, no wait-free linearizable implementation of a
resettable min-CAS register from read, write, and CAS has constant quiescent
complexity.
\end{corollary}

\begin{proof}
Substitution into the sparse-map client gives a read/write/CAS counter with
constant quiescent complexity: at outer quiescence every inner call has
returned, the assumed implementation contributes a constant number of
reachable locations, and $\Owner,\mathit{Reply}$ contribute two.
This contradicts B\'edin et al.\
\cite{BedinEtAl2021}.
\end{proof}

Resettable min-CAS is a restricted generalized compare-and-swap (GCAS):
\[
  \begin{aligned}
  \mathsf{MinUpdate}(x)&=\mathsf{GCAS}(>,x,x),\\
  \mathsf{CAS}(a,b)&=\mathsf{GCAS}(=,a,b).
  \end{aligned}
\]
Prior GCAS work uses FAI, repeated minimum-ticket announcement, and equality
replacement to obtain an analogous $O(\beta_o+1)$ bound for atomic state
transitions
\cite[Appendix~B.2, Theorem~B.2.15]{HadzilacosThiessenToueg2026}.
Its state CAS invalidates stale snapshots; \alg uses an owner generation to
preserve the bound across a multi-step thunk.
For fixed $N$, our result characterizes unbounded-register count.
Adaptive infinite-population space and a CAS-derived \alg time bound remain
open.

\section{Comparison and related work}\label{sec:related}

Generalized compare-and-swap provides the analogous fixed-senior bound for
pure state transitions.
We formulate this bound as an object-level property and extend it to bounded,
concurrently idempotent multi-step thunks; Section~\ref{sec:separation}
compares the primitive requirements.
The cited adaptive universal constructions give point-contention bounds that
may include operations arriving after key assignment
\cite{AfekDauberTouitou1995,FatourouKallimanis2011}, and therefore do not
imply Definition~\ref{def:seniority}.
Large-object constructions instead address state-copying costs
\cite{AndersonMoir1999}.
Our separation uses the lower bounds of Fich, Hendler, and Shavit and of
B\'edin et al.\ \cite{FichHendlerShavit2004,BedinEtAl2021}.
Unlike monotone registers and general reset transformations, the candidate
combines minimum update with exact reset
\cite{AspnesAttiyaCensorHillel2012,AghazadehGolabWoelfel2013}.
\alg adds a ticket guard and versioned owner to descriptor helping and
idempotence \cite{TurekShashaPrakash1992,Barnes1993,BenDavidBlellochWei2022}.
Randomized wait-free locks give probabilistic contention-dependent bounds
\cite{BenDavidBlelloch2022,AeiniEtAl2026}; \alg gives a deterministic
fixed-senior bound for one always-executed thunk.

\section{Conclusion}\label{sec:conclusion}

\alg is a deterministic helpable thunk lock whose worst-case
invoking-process bound is independent of calls invoked after its ticket step.
A restored ticket guard and versioned owner give
$O(\beta_D+1)$ management steps under every asynchronous schedule.
We formulate this guarantee as retrospective wait-freedom, a property defined
for arbitrary concurrent objects.
For deterministic sequential objects satisfying the uniform $T$-step thunk
contract, one \alg instance gives an effectful retrospective universal
construction with $O((\beta_D+1)(T+1))$ local shared-memory steps per
operation.
In the fixed-population unbounded-value black-box model, resettable min-CAS
has tight $\Theta(N)$ read-write-conditional-register complexity.
FAI supplies the increasing tickets used to identify seniors, and the
infinite-arrival result rules out constant read/write/CAS quiescent location
complexity.

\section*{Acknowledgments}

Generative AI tools assisted only with language editing and formatting.
The authors reviewed all output and take full responsibility for this paper.

This appendix uses the notation, definitions, and algorithm numbers introduced
above and gives the full proof details that Sections 2--9 abbreviate.

\appendix

\section{Safety details}\label{app:safety}

\subsection{Critical interleavings}

\paragraph{Stale candidate ownership.}
Let a helper read $\Cand=x$, then let a descriptor $y$ with a smaller key
replace $x$ before the owner CAS.
The helper may still install $x$ in the same empty generation.
Its cleanup compares the exact token $x$ and therefore cannot remove $y$;
after $x$ completes and releases, $y$ remains the next obligation.

\paragraph{Done-test/generation handshake.}
The installation path orders its steps as an empty-owner read in generation
$g$, a false done read, and then owner CAS.
If $x$ becomes done between the latter two steps, some tenure of $x$ must
start and release before the owner is empty again, advancing the generation.
The delayed CAS expecting $g$ fails.
Moving the done read before the owner read would instead allow installation
of a completed descriptor in the new empty generation.

\paragraph{Physical late run.}
A helper can validate $(x,g)$ and pause while another helper completes $x$,
releases $g$, and installs a later owner.
No finite reread sequence closes the final instruction window before
$x.\mathsf{run}()$.
The observational concurrent-idempotence contract makes every such
post-completion physical step removable under every continuation.

\paragraph{Reset before release and stale self-publication.}
After $x$ is done, cleanup may reset $\Cand=x$ and pause before owner release.
A new candidate token different from $x$ survives because all later cleanup
compares against $x$, but it cannot install until the occupied owner is
released; this is the unguarded tail.
Separately, $x$'s requester can have one loop test already read as active
when another helper finishes $x$.
It may publish one stale done token, then observe done at the next loop test;
it cannot publish a second.

\subsection{Safety invariants}

\paragraph{Generation uniqueness.}
All installations CAS the complete owner pair.
Only one CAS can change $(\nullowner,g)$, and the only transition back to an
empty owner writes $(\nullowner,g+1)$.
Thus at most one descriptor is installed in generation $g$, and every
delayed installation expecting $(\nullowner,g)$ fails after release.

\paragraph{Release after completion.}
The only release is the final CAS in \textsc{HelpOwner}.
It is executed only after observing the monotone done predicate, so tenure
$(x,g)$ cannot end before $x$ is done.

\paragraph{Candidate responsibility.}
The cleanup CAS in \textsc{Assist} is reached only after reading $x$ done,
and the reset in \textsc{HelpOwner} is likewise after service and a done
test.
Hence neither clears an active $x$.
Minimum update only decreases the stored key, and an exact-token reset of a
stale larger-key candidate cannot erase a smaller-key replacement.

\paragraph{One tenure per descriptor.}
Suppose $x$ is installed in generation $g$.
Its tenure ends only after $x$ is done and advances the generation.
An installation prepared before release expects a generation at most $g$ and
fails afterward.
An installation prepared afterward reads the owner before testing done and
therefore observes the monotone true predicate instead of attempting CAS.

\begin{proof}[Detailed proof of safety]
Every run starts after validation of an owner tenure.
If that tenure ends before the helper's next thunk step, release after
completion implies that the descriptor is already done; concurrent
idempotence makes the resumed run noneffectual.
Thus every effectual thunk step belongs to the descriptor whose tenure was
validated.
A distinct tenure cannot begin until the previous descriptor has a finished
run.
Concurrent idempotence gives each descriptor one logical execution interval
ending at its first finished run and makes redundant or later runs
noneffectual, so distinct logical intervals do not overlap.

A descriptor is not published before its lock invocation and its caller
returns only after observing done.
Its logical execution interval therefore lies within its call.
Order completed calls by owner generation.
If one call returns before another is invoked, the former logical interval
has ended; the latter cannot begin effectual work before a later tenure.
Generation order consequently respects real-time precedence and yields a
legal serial thunk execution.
\end{proof}

\section{Finite charge-set progress details}\label{app:progress}

Fix an execution $\alpha$, victim $D$, key $d=\rho(D)$, and a charge set
$\mathcal S_D(\alpha)$.
All events below occur in $D$'s pending suffix (which is infinite if $D$
never finishes), and $m=|\mathcal S_D(\alpha)|$.

\paragraph{Guard establishment.}
Minimum update leaves $D$ or a smaller-key descriptor in $\Cand$.
Larger-key updates preserve that value, and candidate CAS operations only
replace an exact token by $\INF$.

\paragraph{Smaller-key chain compression.}
Every effective smaller-key publication strictly decreases the candidate key.
The charge-set definition assigns its descriptor to $\mathcal S_D$, so a
first-time strict chain has length at most $m$.
A repeated publication is already present or dominated by a smaller-key
token and does not extend the chain.
No point in the chain exposes a larger-key token, so the chain does not
create another stale installation cohort.

\paragraph{Opening the guard.}
The guard can be destroyed only by a successful
$\Cand.\mathsf{CAS}(x,\INF)$ after $x$ is done.
If $x\ne D$, guardedness gives $\rho(x)<d$.
Except for the done token already resident when $D$ receives its key, $x$
was initially resident and active or was published later, so the charge-set
definition places it in $\mathcal S_D$.
After completion, one successful cleanup CAS in \textsc{Assist} or
\textsc{HelpOwner} removes $x$.
Its requester can have at most one loop test already read as active and can
therefore create at most one stale post-completion publication.
Thus each charged descriptor opens the guard at most twice.

\paragraph{Done-candidate episodes.}
If $x$ is in $\Cand$ when it becomes done, that event begins the first
episode.
After removal, only the one stale requester iteration just described can
reinsert it.
Every failed cleanup CAS following a read of done $x$ is charged to one of
these two maximal episodes.

\paragraph{One larger-key bypass per guard.}
No larger-key descriptor becomes candidate while the guard holds.
Every such installation was prepared from a candidate read before the
guarded interval.
If a larger-key owner is already installed, it is the one bypass.
Otherwise all prepared installations that can still succeed expect the same
empty generation $g$.
Generation uniqueness permits one success; release advances to $g+1$, and
the guard prevents preparation of another larger-key installation.
The same argument applies after a smaller-key owner present at interval start:
its reset either ends the interval or leaves a guarded candidate.

\paragraph{One larger-key tail bypass.}
Consider $D$'s program counter when the opening occurs.
If $D$ is between iterations, its next arbitration access is minimum update,
so the tail exposes nothing.
If the opening occurs after $D$'s minimum update but before its owner read,
the remaining iteration follows at most the one owner path selected by that
read.
If it previously read an occupied owner, only that one
\textsc{HelpOwner} suffix remains; success services that tenure and failed
revalidation returns immediately.
If it read an empty owner in generation $g$, the remaining \textsc{Assist}
suffix reads one candidate.
It either returns on $\INF$, performs one cleanup after a done read, or makes
one owner CAS after an active read and follows at most that installed owner.
These branches are exclusive.
If $D$ itself opened the guard in cleanup, both \textsc{Assist}'s
done branch and \textsc{HelpOwner}'s release suffix return without another
candidate/owner read.
Thus the suffix exposes exactly zero or one larger-key obligation.
Arbitrarily many larger-key calls may finish while $D$ is paused; only the
single owner or done-candidate state observed when $D$ resumes matters.
If tenure $H$ first invalidates an empty pair read in the tail, $H$ remains
the tail charge even when the CAS executes after a smaller-key publication
starts the next guard; a different owner $B$ then current at guard start is
that guard's one bypass.

\paragraph{Larger-key bypass compression.}
Let $h$ be the number of openings before $D$ is done.
The opening argument gives $h\le2m+1$, including the initial done-candidate
residue.
There are at most $h+1$ guarded intervals and $h$ tails, each with one
larger-key bypass, for at most
\[
  (h+1)+h\le4m+3
\]
larger-key obligations visible to $D$.

\paragraph{Tenures and local service.}
Every descriptor in $\mathcal S_D$ owns at most one tenure, as does $D$.
Every smaller-key tenure beginning after key assignment belongs to
$\mathcal S_D$.
Adding $m$ smaller-key tenures, at most $4m+3$ larger-key bypasses, $D$, and
one initial owner yields at most $5m+5$ victim-visible charges.
If owner revalidation fails, that tenure has ended permanently.
If it succeeds, bounded service returns with done true and the helper
attempts both cleanup CAS operations before its next loop, so the requester
services a tenure at most once.

\paragraph{Terminal suffix.}
After $D$ becomes done, its sequential requester either observes done at its
next loop test or finishes the one iteration whose test was already read.
That suffix has constant management work and at most one service call.

\paragraph{Management-step charging.}
Partition every precompletion iteration by its \textsc{Assist} outcome.
If it reads an occupied owner, its revalidation either fails because the
tenure ended or succeeds and bounded service plus cleanup ends that tenure
before the requester can encounter it again.
If an owner-installation CAS fails after an empty read, charge the iteration
to the tenure that first changed that exact empty generation.
The same tenure can receive this failed-installation charge and at most one
subsequent occupied-owner charge, so tenure multiplicity is at most two.

After $D$ publishes, a read of $\INF$ requires an intervening exact reset and
is charged to that opening.
The requester performs only one candidate read in the iteration and
republishes at the next loop head, so an opening receives at most one such
charge.
If $D$ reads a done candidate, charge the iteration to that maximal
done-candidate episode.
Whether cleanup succeeds or fails because the candidate changes, that episode
ends before $D$ can issue another cleanup, so its multiplicity is one.
Smaller-key descriptors contribute at most $2m$ episodes and the initial
token one.
A larger-key done token read after an opening is the one tail obligation
already charged to that opening, even if its tenure completed entirely while
$D$ was paused.
These cases exhaust the branches after \textsc{Assist}'s owner read.
There are at most $2m+1$ openings and $5m+5$ owner/bypass charges, hence at
most $7m+6$ charge objects.
Each charge object receives only $O(1)$ iterations.
Together with the constant terminal suffix, this proves the management and
service bounds.

\paragraph{FAI senior cover.}
For \alg, set $\rho(x)=x.\mathit{ticket}$ and
$\sigma_D=\tau_D$.
Every unfinished smaller-ticket call at $\tau_D$ is a senior.
Increasing FAI prevents a later smaller ticket; a returned call cannot
publish again, and any delayed installation not already current at
$\tau_D$ expects an obsolete generation.
The one current owner tenure is the allowed initial residue.
Thus $\mathsf{Seniors}(D)$ is a charge set and has size
$\beta_D\le\kappa_D-1$.

\section{Effectful retrospective universality}\label{app:universal}

We give the full safety, completion, and progress argument for the
effectful retrospective universality theorem.
Let $\mathcal O$ be a deterministic sequential object with state space
$Q$, initial state $q_0$, and transition function
$\delta(q,op)=(q',r)$.
Its state has a private shared representation $R$, and every operation is
packaged as a fresh concurrently idempotent thunk satisfying the uniform,
helper-count-independent $T$-step service contract.
All target-object invocations use one \alg instance.

\begin{theorem}[Effectful retrospective universality, restated]
The wrapper
\[
  \mathsf{Invoke}_{\mathcal O}(op)
  =L.\mathsf{Lock}(D,f_{D,op})
\]
is linearizable and retrospective wait-free.
An invocation with seniority $\beta_D$ performs
$O(\beta_D+1)$ lock-management steps and
$O((\beta_D+1)(T+1))$ local shared-memory steps, independently of
operations invoked after its ticket step.
\end{theorem}

\begin{proof}
Fix an arbitrary target-object history.
Complete every pending invocation whose descriptor is done by appending its
recorded response, and discard every other pending invocation.
Call the resulting completed invocations \emph{retained}.

By lock safety, each retained descriptor has a logical thunk execution
interval between its invocation and response.
Logical intervals of distinct descriptors do not overlap, so they are
totally ordered by time, and this order respects the real-time precedence of
nonoverlapping invocations.
For each interval, concurrent idempotence supplies a subsequence of the
physical thunk steps that is indistinguishable under every continuation from
one run of that thunk, while preserving all steps belonging to other
descriptors.
Because $R$ is private to the target implementation, replacing the physical
runs, interval by interval, therefore leaves an execution in which exactly
one copy of each retained sequential operation accesses the target state.

Induct on the logical-interval order.
The first retained thunk starts from the representation of $q_0$ and,
by its sequential specification, changes it to $q_1$ and records response
$r_1$, where
\[
  \delta(q_0,op_1)=(q_1,r_1).
\]
If the first $i-1$ intervals implement the legal prefix ending in $q_{i-1}$,
then no other logical target operation overlaps interval $i$.
The single effectual run in that interval consequently implements
\[
  \delta(q_{i-1},op_i)=(q_i,r_i)
\]
and records $r_i$ before publishing done.
Thus interval order yields a legal sequential history of $\mathcal O$ with
the responses returned by all retained invocations.
Since the order also respects real time, it is a linearization.

It remains to justify discarding the pending descriptors that are not done.
A descriptor that has not begun an owner tenure has executed no target-state
step.
If an unfinished descriptor has begun a tenure, release after completion
prevents that tenure from ending, so no later logical target interval can
begin.
Its partial physical execution is therefore a suffix after every retained
operation; omitting the pending invocation from the completed abstract
history does not change the legal sequential prefix above.
If the descriptor later finishes, concurrent idempotence makes its full
logical run one additional sequential transition, so the argument is also
closed under continuations.

For progress, a target-object invocation performs exactly the corresponding
$\mathsf{Lock}(D,f_{D,op})$ call.
The main retrospective wait-freedom theorem applies with the same descriptor
ticket, seniors, and $T$-step service contract, yielding
$O(\beta_D+1)$ management steps and
$O((\beta_D+1)(T+1))$ total local shared-memory steps.
Every target invocation after $D$'s ticket step obtains a larger FAI ticket,
so it cannot enter $D$'s senior set or affect either bound.
\end{proof}

\section{Ablation schedules}\label{app:ablation}

\begin{proof}[Proof of the Republication proposition]
Let a smaller-ticket call $E$ obtain a ticket and pause before publication.
Victim $D$ obtains the next ticket, publishes once, and pauses.
Publish $E$, let it become owner, and schedule $D$ to complete it.
The exact reset removes $E$ while $D$'s prior token is already gone.
For round $i$, start a larger-ticket call $Y_i$, publish and install it just
before $D$ reads the owner, and let $D$ complete and release it.
Return $Y_i$ before starting $Y_{i+1}$.
The variant never republishes $D$, so the schedule repeats forever with only
$D$ and one other call active.
\end{proof}

\begin{proof}[Proof of the Stale-cohort lower bound]
Use victim $D$, smaller-ticket calls $E_1,\ldots,E_q$, and larger-ticket
helpers
$H_1,\ldots,H_q$.
All smaller-ticket calls obtain tickets and pause before $D$ receives its
ticket.
In phase $i$, publish one remaining $E_i$, make it owner, and let $D$
complete it, opening the guard.
Let $h_i$ be the minimum remaining helper.
Run every remaining helper through an empty-owner read, a read of
$\Cand=h_i$, and a false done read, pausing each before owner CAS.
Then republish $D$.

Resume the paused CAS operations one by one.
The first installs $h_i$; $D$ completes and releases it.
With an untagged empty owner, every remaining stale CAS can reinstall the
same done descriptor and be cleared by $D$.
Phase $i$ therefore contributes $q-i+1$ tenures.
Summing gives
\[
  q+(q-1)+\cdots+1=\frac{q(q+1)}2.
\]
There are at most $2q+1$ active calls.
A generation permits only the first CAS in each phase and invalidates the
rest on release.
\end{proof}

\section{Primitive-separation details}\label{app:separation}

This section gives the omitted proofs and pseudocode for the
primitive-separation results.

\paragraph{Constant-time counter completion.}
While tenure $(X,g)$ is current, no distinct valid tenure begins.
A delayed helper from an earlier tenure expects a reply vector preceding the
successful update required before that tenure's release.
Reply sequence fields only increase, so that expected vector never returns.
While $X$ is unfinished, every successful reply CAS therefore completes
$X$.
If a helper's CAS fails, one reread observes done.
Each success strictly advances one process sequence from $k-1$ to $k$.

\begin{proof}[Proof of the Counter client theorem]
Processes have permanent identifiers and one outstanding call.
Inductively, $\mathit{Reply}[p].\mathit{seq}=k-1$ when call $k$ starts, and
call $k+1$ is not announced before call $k$ responds.
The first successful reply CAS writing $(k,g)$ for process $p$ linearizes
$X=(\rho(p,k),p,k)$.
Generation uniqueness assigns one token to each tenure and release follows
completion.
Starting at generation zero, consecutive tenures complete one operation and
return $0,1,2,\ldots$.
If a caller returns before release, no later operation completes until that
release advances the generation, preserving real time.
The sequencing invariant prevents a later call from making $(p,k)$ appear
done or replacing its result before response.
The done predicate is monotone, completion costs $O(1)$, and the locally
finite smaller-key set has size at most $\rho(X)$; the finite charge-set
theorem gives wait-freedom.
\end{proof}

\begin{proof}[Proof of the Conditional-RMW space-separation theorem]
Let $s(N)$ count the fixed set of all shared base locations, including
read/write-only locations.
Assume an $s(N)$-register wait-free linearizable long-lived resettable
min-CAS implementation using read/write/conditional-RMW registers.
Substitute it into the counter client and retain fresh owner and reply
registers.
Composition gives an $N$-process wait-free long-lived counter using
$s(N)+2$ read-write-conditional registers.
It introduces no helper process, so the process parameter remains $N$.
The Fich--Hendler--Shavit visible-object theorem requires
$\lceil N/2\rceil$ such registers, even with unbounded values.
Hence $s(N)\ge\lceil N/2\rceil-2$.
\end{proof}

\begin{algorithm}[t]
\caption{$N+1$-register resettable min-CAS upper bound.}
\begin{algorithmic}[1]
\Procedure{Invoke}{$op$}
  \State $k_p\gets k_p+1$;\quad
         $X\gets(p,k_p,N(k_p-1)+p,op)$
  \State $\mathit{Ann}[p].\mathsf{Write}(X)$
  \Loop
    \State $S\gets\mathit{State}.\mathsf{Read}()$
    \If{$S.\mathit{Reply}[p].\mathit{seq}\ge k_p$}
      \State \Return $S.\mathit{Reply}[p].\mathit{result}$
    \EndIf
    \For{$q\gets0$ \textbf{to} $N-1$}
      \State $A[q]\gets\mathit{Ann}[q].\mathsf{Read}()$
    \EndFor
    \State $S\gets\mathit{State}.\mathsf{Read}()$
    \State $\mathcal E\gets
      \{A[q]\ne\bot:A[q].\mathit{seq}=S.\mathit{Reply}[q].\mathit{seq}+1\}$
    \If{$\mathcal E\ne\varnothing$}
      \State $Y\gets\arg\min_{Z\in\mathcal E}\rho_N(Z)$
      \State $(C',v)\gets\Call{Apply}{S.C,Y.op}$
      \State $S'\gets S$ with $S'.C\gets C'$ and
             $S'.\mathit{Reply}[Y.p]\gets(Y.\mathit{seq},v)$
      \State $\mathit{State}.\mathsf{CAS}(S,S')$
    \EndIf
  \EndLoop
\EndProcedure
\end{algorithmic}
\end{algorithm}

\begin{proof}[Proof of the Linear CAS upper-bound theorem]
\leavevmode\par
The packed state is $(C,\mathit{Reply}[0..N-1])$.
Each successful state CAS applies one eligible operation to the preceding
$C$ and advances its owner's reply sequence by exactly one.
It is the operation's linearization point.
Sequence monotonicity prevents duplicate application and logical ABA of the
unbounded packed state even when $C$ cycles.
A stale announcement is filtered if complete at the following state read; if
it completes afterward, the state changes and the stale CAS fails.
Thus the successful CAS chain is a legal sequential history and every
completed operation linearizes between publication and response.

Fix pending $X$ and its announcement write $W_X$.
At $W_X$, each other sequential process has at most one help call that
already read $\mathit{Ann}[X.p]$.
Each makes one CAS attempt, so at most $N-1$ successful transitions after
$W_X$ can arise from $X$-blind collects.
Every new collect sees $X$ until completion.
If its state read still sees $X$ pending, minimum-key choice selects $X$ or
one of the finitely many operations in $L_X^N$.
Every successful CAS completes a distinct identity.
Each victim CAS failure is paired with a distinct successful transition
between its state read and CAS.
Therefore the requester performs
$O(|L_X^N|+N)$ iterations while pending, followed by one reply read, each
using $N+O(1)$ accesses.
Bounded sequence wraparound and pointer reclamation are outside this abstract
theorem.
\end{proof}

\paragraph{Infinite arrivals.}
With permanent identifiers in $\mathbb N$, encode replies as one flat
by-value sparse map in an unbounded register.
Entries are monotone and never removed; operation tokens contain
$(\rho,p,k)$ values rather than retained descriptor references.
A constant-quiescent-complexity read/write/CAS implementation of resettable
min-CAS would then compose with the two-register client state.
At outer quiescence all nested calls have returned, so reachable locations
are exactly the implementation's constant quiescent set plus owner and reply.
This contradicts the cited infinite-arrival counter impossibility.


\end{document}